\documentclass[aps,superscriptaddress,11pt,twoside]{revtex4}

\usepackage{amsmath,latexsym,amssymb,verbatim,enumerate,graphicx}
\usepackage[all]{xy}

\usepackage{theorem}
\newtheorem{definition}{Definition}[section]

\newtheorem{lemma}[definition]{Lemma}

\newtheorem{theorem}[definition]{Theorem}

\def\squareforqed{\hbox{\rlap{$\sqcap$}$\sqcup$}}
\def\qed{\ifmmode\squareforqed\else{\unskip\nobreak\hfil
\penalty50\hskip1em\null\nobreak\hfil\squareforqed
\parfillskip=0pt\finalhyphendemerits=0\endgraf}\fi}
\def\endenv{\ifmmode\;\else{\unskip\nobreak\hfil
\penalty50\hskip1em\null\nobreak\hfil\;
\parfillskip=0pt\finalhyphendemerits=0\endgraf}\fi}
\newenvironment{proof}{\noindent \textbf{{Proof~} }}{\qed}

\mathchardef\ordinarycolon\mathcode`\:
\mathcode`\:=\string"8000
\def\vcentcolon{\mathrel{\mathop\ordinarycolon}}
\begingroup \catcode`\:=\active
  \lowercase{\endgroup
  \let :\vcentcolon
  }

\newcommand{\nc}{\newcommand}
\nc{\rnc}{\renewcommand} \nc{\beq}{\begin{equation}}
\nc{\eeq}{{\end{equation}}} \nc{\bea}{\begin{eqnarray}}
\nc{\eea}{\end{eqnarray}} \nc{\beqa}{\begin{eqnarray}}
\nc{\eeqa}{\end{eqnarray}} \nc{\lbar}[1]{\overline{#1}}
\nc{\bra}[1]{\langle#1|} \nc{\ket}[1]{|#1\rangle}
\nc{\ketbra}[2]{|#1\rangle\!\langle#2|}
\nc{\braket}[2]{\langle#1|#2\rangle} \nc{\proj}[1]{|
#1\rangle\!\langle #1 |} \nc{\avg}[1]{\langle#1\rangle}
\rnc{\max}{\operatorname{max}} \nc{\rank}{\operatorname{rank}}
\nc{\conv}{\operatorname{conv}}
\nc{\smfrac}[2]{\mbox{$\frac{#1}{#2}$}} \nc{\Tr}{\operatorname{Tr}}
\nc{\ox}{\otimes} \nc{\dg}{\dagger} \nc{\dn}{\downarrow}
\nc{\cA}{{\cal A}} \nc{\cB}{{\cal B}} \nc{\cC}{{\cal C}}
\nc{\cD}{{\cal D}} \nc{\cE}{{\cal E}} \nc{\cF}{{\cal F}}
\nc{\cG}{{\cal G}} \nc{\cH}{{\cal H}} \nc{\cI}{{\cal I}}
\nc{\cJ}{{\cal J}} \nc{\cK}{{\cal K}} \nc{\cL}{{\cal L}}
\nc{\cM}{{\cal M}} \nc{\cN}{{\cal N}} \nc{\cO}{{\cal O}}
\nc{\cP}{{\cal P}} \nc{\cR}{{\cal R}} \nc{\cS}{{\cal S}}
\nc{\cT}{{\cal T}} \nc{\cU}{{\cal U}} \nc{\cV}{{\cal V}}
\nc{\cX}{{\cal X}} \nc{\cW}{{\cal W}} \nc{\cZ}{{\cal Z}}
\nc{\csupp}{{\operatorname{csupp}}}
\nc{\qsupp}{{\operatorname{qsupp}}} \nc{\rar}{\rightarrow}
\nc{\lrar}{\longrightarrow} \nc{\poly}{\operatorname{poly}}
\nc{\polylog}{\operatorname{polylog}} \nc{\Lip}{\operatorname{Lip}}
\nc{\supp}{{\operatorname{supp}}}

\def\>{\rangle}
\def\<{\langle}

\def\a{\alpha}
\def\b{\beta}

\def\d{\delta}
\def\e{\epsilon}

\def\l{\lambda}

\def\r{\rho}
\def\s{\sigma}

\def\ph{\varphi}

\def\ps{\psi}

\def\G{\Gamma}

\def\Ph{\Phi}

\def\O{\Omega}

\nc{\glneq}{{\raisebox{0.6ex}{$>$}  \hspace*{-1.8ex} \raisebox{-0.6ex}{$<$}}}
\nc{\gleq}{{\raisebox{0.6ex}{$\geq$}\hspace*{-1.8ex} \raisebox{-0.6ex}{$\leq$}}}

\nc{\RR}{{{\mathbb R}}}
\nc{\CC}{{{\mathbb C}}}
\nc{\FF}{{{\mathbb F}}}
\nc{\HH}{{{\mathbb H}}}
\nc{\NN}{{{\mathbb N}}}
\nc{\ZZ}{{{\mathbb Z}}}
\nc{\PP}{{{\mathbb P}}}
\nc{\QQ}{{{\mathbb Q}}}
\nc{\UU}{{{\mathbb U}}}
\nc{\WW}{{{\mathbb W}}}
\nc{\EE}{{{\mathbb E}}}
\rnc{\SS}{{{\mathbb S}}}
\nc{\id}{{\operatorname{id}}}

\nc{\vholder}[1]{\rule{0pt}{#1}}

\nc{\ob}[1]{#1}

\def\beq{\begin {equation}}
\def\eeq{\end {equation}}

\nc{\eq}[1]{Eq.~(\ref{eq:#1})} \nc{\eqs}[2]{Eqs.~(\ref{eq:#1}) and
(\ref{eq:#2})}

\nc{\eqn}[1]{Eq.~(\ref{eqn:#1})}
\nc{\eqns}[2]{Eqs.~(\ref{eqn:#1}) and (\ref{eqn:#2})}

\nc{\Hm}{{H^{min}}} \nc{\Hpm}{{H_p^{min}}} \rnc{\log}{\ln}
\nc{\Wg}{\operatorname{Wg}}

\begin{document}

\title{{\Large The maximal $p$-norm multiplicativity conjecture is false}}

\author{Patrick Hayden}
 \email{patrick@cs.mcgill.ca}
 \affiliation{
    School of Computer Science,
    McGill University,
    Montreal, Quebec, H3A 2A7, Canada
    }

\date{July 22, 2007}

\begin{abstract}
For all $1 < p < 2$, we demonstrate the existence of quantum
channels with non-multiplicative maximal $p$-norms. Equivalently,
the minimum output Renyi entropy of order $p$ of a quantum channel
is not additive for all $1 < p < 2$. The violations found are large.
As $p$ approaches 1, the minimum output Renyi entropy of order $p$
for a product channel need not be significantly greater than the
minimum output entropy of its individual factors. Since $p=1$
corresponds to the von Neumann entropy, these counterexamples
demonstrate that if the additivity conjecture of quantum information
theory is true, it cannot be proved as a consequence of maximal
$p$-norm multiplicativity.
\end{abstract}

\maketitle

\parskip .75ex

%%%%%%%%%%%%%%%%%%%%%%%%%%%%%%%%%%%%%%%%%%%%%%%%%%%%%%%%%%%%%%%%%%%%%%%%

\section{Introduction} \label{sec:intro}

The oldest problem of quantum information theory is arguably to
determine the capacity of a quantum-mechanical communications
channel for carrying information, specifically ``classical''
\emph{bits} of information. (Until the 1990's it would have been
unnecessary to add that additional qualification, but today the
field is equally concerned with other forms of information like
\emph{qubits} and \emph{ebits} that are fundamentally
quantum-mechanical.) The classical capacity problem long predates
the invention of quantum source coding~\cite{Schumacher95,JozsaS94}
and was of concern to the founders of information theory
themselves~\cite{P73}.
\begin{comment}
Indeed, in 1973, John Pierce ended his contribution to a 25 year
retrospective on information theory with this now famous
provocation~\cite{P73}:
\begin{quotation}
I think that I have never met a physicist who understood information
theory. I wish that physicists would stop talking about
reformulating information theory and would give us a general
expression for the capacity of a channel with quantum effects taken
into account rather than a number of special cases.
\end{quotation}
\end{comment}
The first major result on the problem came with the resolution of a
conjecture of Gordon's~\cite{G64} by Alexander Holevo in 1973, when
he published the first proof~\cite{H73} that the
maximum amount of information that can be extracted from an ensemble
of states $\r_i$ occurring with probabilities $p_i$ is bounded above
by
\begin{equation}
\chi(\{p_i,\r_i\})
 = H\Big( \sum_i p_i \r_i \Big) - \sum_i p_i H(\r_i),
\end{equation}
where $H(\r) = -\Tr \r \ln \r$ is the von Neumann entropy of the
density operator $\r$. For a quantum channel $\cN$, one can then
define the Holevo capacity
\begin{equation}
\chi(\cN) = \max_{\{p_i,\r_i\}} \chi( \{ p_i, \cN(\r_i) \} ),
\end{equation}
where the maximization is over all ensembles of input states.
Writing $C(\cN)$ for the classical capacity of the channel $\cN$,
this leads easily to an upper bound of
\begin{equation} \label{eqn:c.chi}
C(\cN) \leq \lim_{n \rar \infty} \frac{1}{n} \chi(\cN^{\ox n}).
\end{equation}
It then took more than two decades for further substantial progress
to be made on the problem, but in 1996, building on recent
advances~\cite{HausladenJSWW96}, Holevo~\cite{H98} and
Schumacher-Westmoreland~\cite{SW97} managed to show that the upper
bound in Eq.~(\ref{eqn:c.chi}) is actually achieved. This was a
resolution of sorts to the capacity problem, but the limit in the
equation makes it in practice extremely difficult to evaluate. If
the codewords used for data transmission are restricted such that
they are not entangled across multiple uses of the channel, however,
the resulting \emph{product state capacity} $C_{1\infty}(\cN)$ has
the simpler expression
\begin{equation}
C_{1\infty}(\cN) = \chi(\cN).
\end{equation}
The additivity conjecture for the Holevo capacity asserts that
for all channels $\cN_1$ and $\cN_2$,
\begin{equation}
\chi(\cN_1 \ox \cN_2) = \chi(\cN_1) + \chi(\cN_2).
\end{equation}
This would imply, in particular, that $C_{1\infty}(\cN) = C(\cN)$,
or that entangled codewords do not increase the classical capacity
of a quantum channel.

\begin{comment}
In the course of the development of quantum information theory,
questions also arose about the additivity of a number of other
important quantities. For a bipartite quantum state $\r^{AB}$, the
\emph{entanglement of formation}~\cite{BennettDSW96} is defined to
be
\begin{equation} \label{eqn:e.o.f}
E_f(\r^{AB}) = \min_{\{p_i,\ket{\ph_i}\}} \sum_i p_i H\Big( \Tr_B
\proj{\ph_i}^{AB} \Big),
\end{equation}
where the minimization is over all pure state decompositions $\sum_i
p_i \proj{\ph_i}$ of the density operator $\r$. The minimal rate at
which of singlets must be consumed in order to manufacture many
copies of $\r$ if the holders of the $A$ and $B$ systems are
restricted to exchanging bits and acting locally is the
\emph{entanglement cost}~\cite{HaydenHT01}, which satisfies
\begin{equation}
E_c(\r) = \lim_{n \rar \infty} \frac{1}{n} E_f(\r).
\end{equation}
The ``weak'' version of the additivity conjecture for the
entanglement of formation states that $E_c(\r) = E_f(\r)$ and there
is an analogous stronger version for pairs of states $\r^{AB}$ and
$\s^{AB}$. As one of two ``extremal'' measures of
entanglement~\cite{HorodeckiHH00}, the other being the entanglement
of distillation, the entanglement of formation plays a central role
in the theory of mixed state entanglement and a great deal of effort
has been devoted to evaluating
it~\cite{Wootters98,VW01,VidalDC02,MatsumotoY04}.
\end{comment}

In 2003, Peter Shor~\cite{Sh02}, building on several previously
established
connections~\cite{Pomeransky03,AudenaertB04,MatsumotoSW04},
demonstrated that the additivity of the Holevo capacity, the
additivity of the entanglement of
formation~\cite{BennettDSW96,HaydenHT01,VidalDC02,MatsumotoY04} and
the superadditivity of the entanglement of formation~\cite{VW01} are
all equivalent to yet another conjecture, known as the \emph{minimum
entropy output conjecture}~\cite{KingR01}, which is particularly
simple to express mathematically.  For a channel $\cN$, define
\begin{equation} \label{eqn:min.entropy}
\Hm(\cN) = \min_{\ket{\ph}} H( \cN( \ph ) ),
\end{equation}
where the minimization is over all pure input states $\ket{\ph}$.
The minimum entropy output conjecture asserts that for all channels
$\cN_1$ and $\cN_2$,
\begin{equation} \label{eqn:add}
\Hm(\cN_1 \ox \cN_2) = \Hm(\cN_1) + \Hm(\cN_2).
\end{equation}

There has been a great deal of previous work on these conjectures,
particularly numerical searches for counterexamples, necessarily in
low dimension, at Caltech, IBM, IMaPh and by ERATO
researchers~\cite{OsawaN01}, as well as proofs of many special
cases. For example, the minimum entropy output conjecture has been
shown to hold if one of the channels is the identity
channel~\cite{AmosovHW00,AmosovH01}, a unital qubit
channel~\cite{K01}, a generalized depolarizing
channel~\cite{FujiwaraH02,K03} or an entanglement-breaking
channel~\cite{Holevo98a,K01b,Sh02}. In addition, the weak additivity
conjecture was confirmed for degradable channels~\cite{DevetakS05},
their conjugate channels~\cite{KingMNR05} and some other special
classes of
channels~\cite{Cortese04,MatsumotoY04,DattaHS04,Fukuda05}. Further
evidence for qubit channels was supplied in~\cite{KingR01}. This
list is by no means exhaustive. The reader is directed to Holevo's
reviews for a detailed account of the history of the additivity
problem~\cite{Holevo04,Holevo07}.

For the past several years, the most commonly used strategy for
proving these partial results has been to demonstrate the
multiplicativity of maximal $p$-norms of quantum channels for $p$
approaching 1~\cite{AmosovHW00}. For a quantum channel $\cN$ and $p
> 1$, define the maximal $p$-norm of $\cN$ to be
\begin{equation}
\nu_p(\cN) = \sup\Big\{ \big\| \cN(\r) \big\|_p \, ; \r \geq 0, \,
\Tr \r =1 \Big\}.
\end{equation}
In the equation, $\| \s \|_p = \big( \Tr |\s|^p \big)^{1/p}$. The
\emph{maximal $p$-norm multiplicativity
conjecture}~\cite{AmosovHW00} asserts that for all quantum channels
$\cN_1$ and $\cN_2$,
\begin{equation} \label{eqn:p.norm.conjecture}
\nu_p(\cN_1 \ox \cN_2) = \nu_p(\cN_1)\nu_p(\cN_2).
\end{equation}
This can
be re-expressed in an equivalent form more convenient to us using Renyi
entropies. Define the Renyi entropy of order $p$ to be
\begin{equation}
H_p(\rho) = \frac{1}{1-p} \log \Tr \rho^p
\end{equation}
for $p>0$, $p \neq 1$. Since $\lim_{p \downarrow 1} H_p(\rho) =
H(\r)$, we will also define $H_1(\r)$ to be $H(\r)$. All these
entropies have the property that they are 0 for pure states and
achieve their maximum value of the logarithm of the dimension on
maximally mixed states. Define the minimum output Renyi entropy
$\Hpm$ by substituting $H_p$ for $H$ in Eq.~(\ref{eqn:min.entropy}).
Eq.~(\ref{eqn:p.norm.conjecture}) can then be written equivalently
as
\begin{equation} \label{eqn:p.add}
\Hpm(\cN_1 \ox \cN_2) = \Hpm(\cN_1) + \Hpm(\cN_2),
\end{equation}
which underscores the fact that the maximal $p$-norm
multiplicativity conjecture is a natural strengthening of the
original minimum entropy output conjecture (\ref{eqn:add}).

This conjecture spawned a significant literature of its own which we
will not attempt to summarize. Holevo's reviews are again an
excellent source~\cite{Holevo04,Holevo07}. Some more recent
important references
include~\cite{KingR04,KingR05,SerafiniEW05,DevetakJKR06,Michalakis07}.
Unlike the von Neumann entropy case, however, some counterexamples
had already been found prior to this paper. Namely, Werner and
Holevo found a counterexample to Eq. (\ref{eqn:p.add}) for $p >
4.79$~\cite{WH02} that nonetheless doesn't violate the $p$-norm
multiplicativity conjecture for $1 < p <2$~\cite{Datta04}, and very
recently Winter showed that the conjecture is false for all $p >
2$~\cite{W07}. In light of these developments, the standing
conjecture was that the maximal $p$-norm multiplicativity held for
$1 \leq p \leq 2$, corresponding to the region in which the map $X
\mapsto X^p$ is operator convex~\cite{KingR04}. More conservatively,
it was conjectured to hold at least in an open interval $(1,1+\e)$,
which would be sufficient to imply the minimum entropy output
conjecture.

On the contrary, we will show that the conjecture is false for all
$1 < p < 2$. In particular, given $1 < p < 2$, we show that there
exist channels $\cN_1$ and $\cN_2$ with output dimension $d$ such
that both $\Hpm(\cN_1)$ and $\Hpm(\cN_2)$ are equal to $\log d -
\cO(1)$ but $\Hpm( \cN_1 \ox \cN_2 ) = p \log d + \cO(1)$. Thus,
\begin{equation}
 \Hpm( \cN_1 ) + \Hpm( \cN_2 ) - \Hpm( \cN_1 \ox \cN_2 )
 = (2-p) \log d - \cO(1).
\end{equation}
For $p$ close to 1, one finds that the minimum entropy output of the
product channel need not be significantly larger than the minimum
output entropy of the individual factors.
Since~\cite{AmosovHW00,K03}
\begin{equation} \label{eqn:king}
\Hpm( \cN_1 \ox \cN_2 ) \geq \Hpm( \cN_1 ) = \log d - \cO(1),
\end{equation}
these counterexamples are
essentially the strongest possible for $p$ close to $1$.

At $p=1$ itself, however, we see no evidence of a violation of the
additivity conjecture for the channels we study. Thus, the
conjecture stands and it is still an open question whether entangled
codewords can increase the classical capacity of a quantum channel.

\begin{comment}
{\bf Structure of the paper:} The counterexamples to the maximal
$p$-norm multiplicativity conjecture are presented in section
\ref{sec:counterexamples}. The case of the von Neumann entropy,
$p=1$, is discussed in section \ref{sec:vn.entropy}. An appendix
describes a lengthy calculation helpful for understanding what is
happening at $p=1$.
\end{comment}

{\bf Notation:} If $A$ and $B$ are finite dimensional Hilbert
spaces, we write $AB \equiv A\otimes B$ for their tensor product and
$|A|$ for $\dim A$. The Hilbert spaces on which linear operators act
will be denoted by a superscript.  For instance, we write $\ph^{AB}$
for a density operator on $AB$. Partial traces will be abbreviated
by omitting superscripts, such as $\ph^A \equiv \Tr_B\ph^{AB}$.  We
use a similar notation for pure states, e.g.\ $\ket{\psi}^{AB}\in
AB$, while abbreviating $\psi^{AB} \equiv \proj{\psi}^{AB}$. We
associate to any two isomorphic Hilbert spaces $A\simeq A'$ a
\emph{unique} maximally entangled state which we denote
$\ket{\Phi}^{AA'}$.  Given any orthonormal basis $\{\ket{i}^A\}$ for
$A$,  if we define $\ket{i}^{A'} = V\ket{i}^A$ where $V$ is the
associated isomorphism, we  can write this state as $\ket{\Ph}^{AA'}
= |A|^{-1/2}\sum_{i=1}^{|A|}\ket{i}^A\ket{i}^{A'}$.
We will also make use of the asymptotic notation $f(n) = \cO(g(n))$
if there exists $C > 0$ such that for sufficiently large $n$,
$|f(n)| \leq C g(n)$. $f(n) = \Omega(g(n))$ is defined similarly but
with the reverse inequality $|f(n)| \geq C g(n)$. Finally, $f(n) =
\Theta(g(n))$ if $f(n) = \cO(g(n))$ and $f(n) = \O(g(n))$.

\section{The counterexamples} \label{sec:counterexamples}

Let $E$, $F$ and $G$ be finite dimensional quantum systems, then
define $R=E$, $S=FG$, $A=EF$ and $B=G$, so that $RS = AB = EFG$.
\begin{comment}
\begin{equation}
 \begin{array}{r}
 R \Big\{ E  \\
 S \left\{ \begin{array}{c} F \\ G \end{array} \right.
 \end{array}
 \begin{array}{l}
 \left. \begin{array}{l}  \\  \end{array} \right\} \\
 \}
 \end{array}
\end{equation}
\end{comment}
Our counterexamples will be channels from $S$ to $A$ of the form
\begin{equation} \label{eqn:the.channel}
\cN( \r ) = \Tr_B \big[ U ( \proj{0}^R \ox \r ) U^\dg \big]
\end{equation}
for $U$ unitary and $\ket{0}$ some fixed state on $R$. Our method
will be to fix the dimensions of the systems involved, select $U$ at
random, and show that the resulting channel is likely to violate
additivity. The rough intuition motivating our examples will be to
exploit the fact that there are channels that appear to be highly
depolarizing for product state inputs despite the fact that they are
not close to the depolarizing channel in, for example, the norm of
complete boundedness~\cite{Paulsen86}.

Consider a single copy of $\cN$ and the associated map $\ket{0}^R
\ket{\ph}^S \mapsto U \ket{0}^R \ket{\ph}^S$. This map takes $S$ to
a subspace of $A \ox B$, and if $U$ is selected according to the
Haar measure, then the image of $S$ is itself a random subspace,
distributed according to the unitarily invariant measure. In
\cite{HLW06}, it was shown that if $|S|$ is chosen appropriately,
then the image is likely to contain only almost maximally entangled
states, as measured by the entropy of entanglement. After tracing
over $B$, this entropy of entanglement becomes the entropy of the
output state. Thus, for $S$ of suitable size, all input states get
mapped to high entropy output states. We will repeat the analysis
below, finding that the maximum allowable size of $S$ will depend on
$p$ as described by the following lemma:
\begin{lemma} \label{lem:size.subspace}
Let $A$ and $B$ be quantum systems with $2 \leq |A| \leq |B|$ and $1 < p
< 2$. Then there exists a subspace $S \subset A \ox B$ of dimension
\begin{equation}
 |S| =  \left\lfloor
            \frac{ \G_p|A|^{2-p}|B| \a^{2.5}}{p-1}
        \right\rfloor,
\end{equation}
with $\G_p > 0$ a constant, that contains only states $\ket{\ph} \in
S$ with high entanglement, in the sense that
\begin{equation} \label{eqn:good.entropy}
H_p(\ph^A) \geq \log |A| - \a - \b,
\end{equation}
where $\b = 2|A| / |B|$. The probability that a subspace of
dimension $|S|$ chosen at random according to the unitarily
invariant measure will not have this property is bounded above by
\begin{equation}
 \left( \frac{|A|^{(p-1)/2}}{\a} \right)^{2|S|}
    \exp \left( - \frac{ (2 |A||B|-1)\a^2}{2|A|^{p-1}} \right).
\end{equation}
%Moreover, for $p = 1$, the probability that a random subspace does not obey (\ref{eqn:good.entropy}) is bounded above by
%\begin{equation} \label{eqn:subspace.prob}
%\left( \frac{10 \sqrt{1 + (\ln |A|)^2}}{\a} \right)^{2|S|}
%    \exp \left( - \frac{(2|A||B|-1)\a^2}{8(1 + (\ln |A|)^2)} \right).
%\end{equation}
\end{lemma}
\begin{proof}
The argument is nearly identical to the proof of Theorem IV.1 in
\cite{HLW06} so we will only discuss the differences here, referring
the reader to the original paper to complete the argument. The first
ingredient is the estimate $\EE H_p( \ph^A ) \geq \EE H_2( \ph^A)
\geq \ln |A| - |A|/|B|$, where the expectation is over random pure
states on $A\ox B$~\cite{L78}. All that is required in addition is
an upper bound on the Lipschitz constant of the maps $f_p(\ket{\ph})
= H_p( \ph^A )$ for $p>1$. Let $\ket{\ph} = \sum_{jk} \ph_{jk}
\ket{jk}^{AB}$. We begin, as in \cite{HLW06}, by bounding the
Lipschitz constant of $g_p(\ket{\ph}) = H_p( \sum_j \bra{j} \ph^A
\ket{j} )$. For $q_j = \sum_k |\ph_{jk}|^2$,
\begin{eqnarray}
 \nabla g_p \cdot \nabla g_p
 &=& \frac{4 p^2}{(1-p)^2} \frac{\sum_j q_j^{2p-1}}{(\sum_j q_j^p)^2}
 \leq \frac{4 p^2}{(1-p)^2} \frac{1}{\sum_j q_j^p}
 \leq \frac{4 p^2}{(1-p)^2} |A|^{p-1} \label{eqn:lipschitz.p},
\end{eqnarray}
using the facts that $q_j^{2p-1} \leq q_j^p$ for $p > 1$ in the
first inequality and that $\sum_j q_j^p$ is minimized by the uniform
distribution in the second. Eq.~(\ref{eqn:lipschitz.p}) therefore
provides an upper bound on the square of the Lipschitz constant of
$g_p$. The Schur concavity of the Renyi entropies then ensures that
the same argument as was used in~\cite{HLW06} can be used to upper
bound the Lipschitz constant of $f_p$ by that of $g_p$.
\end{proof}

Now consider the product channel $\cN \ox \bar{\cN}$, where
$\bar\cN(\r) = \Tr_B \big[ \bar U (\proj{0} \ox \r) U^T \big]$. We
will exploit a form of the same symmetry as
Werner-Holevo~\cite{WH02} and Winter~\cite{W07} did, but instead of
using an exact symmetry that occurs only rarely, we'll use an
approximate version of it that holds always. In the trivial case
where $|R| = 1$, the identity $U \ox \bar{U} \ket{\Ph} = (U
\bar{U}^T \ox I) \ket{\Ph} = \ket{\Ph}$ for the maximally entangled
state $\ket{\Ph}^{S_1 S_2}$ implies that
\begin{equation}
(\cN \ox \bar{\cN})(\proj{\Ph}^{S_1 S_2})
 = \Tr_{B_1 B_2} \left[ \proj{\Ph}^{A_1 A_2} \ox \proj{\Ph}^{B_1 B_2} \right]
 = \proj{\Ph}^{A_1 A_2}.
\end{equation}
%where $\ket{\Ph}^{A_1 B_1 A_2 B_2}$ is the image of $\ket{\Ph}^{S_1
%S_2}$ under the isomorphism induced by $S \cong A \ox B$. If we
%choose orthonormal bases for $S$, $A$ and $B$ as well as an
%isomorphism such that $\ket{ij}^S \mapsto \ket{i}^A \ket{j}^B$, then
%$\ket{\Ph}^{A_1 B_1 A_2 B_2}$ will map to a pair of maximally
%entangled states, one on $A_1 \ox A_2$ and another on $B_1 \ox B_2$.
The output of $\cN \ox \bar{\cN}$ will thus be a pure state. In the
general case, we will choose $R$ to be small but not trivial, in
which case useful bounds can still be placed on the largest
eigenvalue of the output state for an input state maximally
entangled between $S_1$ and $S_2$.
\begin{lemma} \label{lem:big.eigenvalue}
Let $\ket{\Ph}^{S_1 S_2}$ be a state maximally entangled between
$S_1$ and $S_2$ as in the previous paragraph. Then $(\cN \ox
\bar\cN)(\Ph^{S_1S_2})$ has an eigenvalue of at least
$\smfrac{|S|}{|A||B|}$.
\end{lemma}
\begin{proof}
This is an easy calculation again exploiting the $U\ox\bar U$
invariance of the maximally entangled state. Recall that $R=E$, $S =
FG$, $A=EF$ and $B=G$:
\begin{align}
\,^{A_1 A_2} \bra{\Ph} & (\cN \ox \bar\cN )(\Ph^{S_1S_2})
\ket{\Ph}^{A_1 A_2} \\
 &= \Tr\left[ (\Ph^{A_1A_2} \ox I^{B_1 B_2})( U \ox \bar U )
    (\proj{00}^{R_1R_2} \ox \Ph^{S_1 S_2})(U^\dg \ox U^T) \right] \\
 &\geq \Tr\left[ (U^\dg \ox U^T) (\Ph^{E_1E_2} \ox \Ph^{F_1F_2} \ox \Ph^{G_1 G_2})( U \ox \bar U )
    (\proj{00}^{E_1E_2} \ox \Ph^{F_1 F_2} \ox \Ph^{G_1G_2}) \right] \\
 &= \Tr\left[ (\Ph^{E_1E_2} \ox \Ph^{F_1F_2} \ox \Ph^{G_1 G_2})
    (\proj{00}^{E_1E_2} \ox \Ph^{F_1 F_2} \ox \Ph^{G_1G_2}) \right]
    \\
 &= \frac{1}{|E|} = \frac{|S|}{|A||B|}.
\end{align}
Note that $U$ acts on $E_1 F_1 G_1$ and $\bar{U}$ on $E_2 F_2 G_2$.
In the third line we have used the operator inequality $\Ph^{G_1G_2}
\leq I^{G_1G_2}$ and the cyclic property of the trace.
\end{proof}

In order to demonstrate violations of additivity, the first step is
to bound the minimum output entropy from below for a single copy of
the channel. Fix $1 < p < 2$, let $|B| = 2|A|$ so that $\b = 1$, set
$\a=1$, and then choose $|S|$ according to Lemma
\ref{lem:size.subspace}. With probability approaching 1 as $|A| \rar
\infty$, when $U$ is chosen according to the Haar measure,
\begin{equation} \label{eqn:Hpm.1copy}
\Hpm( \cN ) \geq \ln |A| - 2.
\end{equation}
The same obviously holds for $\Hpm( \bar\cN )$. Recall that the
entropy of the uniform distribution is $\ln |A|$ so the minimum
entropy is nearly maximal.

On the other hand, by Lemma \ref{lem:big.eigenvalue},
\begin{equation} \label{eqn:renyi.bound}
H_p\big( ( \cN \ox \bar\cN)(\Ph) )
    \leq \frac{1}{1-p} \ln \left(\frac{|S|}{|A||B|}\right)^p
    = \frac{p}{1-p} \ln \frac{|S|}{|A||B|}.
\end{equation}
Substituting the same value of $|S|$ into this inequality yields
\begin{equation} \label{eqn:Hpm.2copies}
 H_p\big( (\cN \ox \bar\cN)(\Ph) \big)
 \leq
 p \log |A| + \cO(1).
\end{equation}
Since $p < 2$, the Renyi entropy of $(\cN \ox \bar\cN)(\Ph)$ is
strictly less than $\Hpm(\cN) + \Hpm(\bar\cN) \geq 2 \ln |A| -
\cO(1)$, where the last inequality holds with high probability. This
is a violation of conjecture (\ref{eqn:p.add}), with the size of the
gap approaching $\ln |A| - \cO(1)$ as $p$ tends to 1.

%Since a theorem seems to be in order:
\begin{theorem} \label{thm:p.counterexamples}
For all $1 < p < 2$, there exists a quantum channel for which the
inequalities (\ref{eqn:Hpm.1copy}) and (\ref{eqn:Hpm.2copies}) both
hold. The inequalities are inconsistent with the maximal $p$-norm
multiplicativity conjecture.
%(\ref{eqn:p.add}).
\end{theorem}

Note, however, that changing $p$ also requires changing $|S|$
according Lemma \ref{lem:size.subspace}, so we have a sequence of
channels violating additivity of the minimal output Renyi entropy as
$p$ decreases to 1, as opposed to a single channel doing so for
every $p$. This prevents us from drawing conclusions about the von
Neumann entropy by taking the limit $p \rar 1$.

As an aside, it is interesting to observe that violating maximal
$p$-norm multiplicativity has structural consequences for the
channels themselves. For example, because entanglement-breaking
channels do not violate multiplicativity~\cite{King03b}, there must
be states $\ket{\ps}^{S_1 S_2}$ such that $(\cN \ox I^{S_2})(\ps)$
is entangled, despite the fact that $\cN$ will be a rather noisy
channel. (The same conclusions apply to the channels used as
examples by Winter~\cite{W07}, where the conclusion takes the form
that $\e$-randomizing maps need not be entanglement-breaking.)

\section{The von Neumann entropy case} \label{sec:vn.entropy}

Despite the large violations found for $p$ close to 1, the class of
examples presented here do not appear to contradict the minimum
entropy output conjecture for the von Neumann entropy. The reason is
that the upper bound demonstrated for
$H_p\big((\cN\ox\bar\cN)(\Ph)\big)$ in the previous section rested
entirely on the existence of one large eigenvalue for
$(\cN\ox\bar\cN)(\Ph)$. The von Neumann entropy is not as sensitive
to the value of a single eigenvalue as are the Renyi entropies for
$p > 1$ and, consequently, does not appear to exhibit additivity
violations. With a bit of work, it is possible to make these
observations more rigorous.
\begin{lemma} \label{lem:avg.purity}
Let $\ket{\Ph}^{S_1 S_2}$ be a maximally entangled state between
$S_1$ and $S_2$. Assuming that $|A| \leq |B| \leq |S|$,
\begin{equation} \label{eqn:avg.purity}
 \int \Tr\Big[ \big( (\cN \ox \bar{\cN})(\proj{\Ph}) \big)^2 \Big] \, dU
    = \frac{|S|^2}{|A|^2|B|^2} + \cO\left(\frac{1}{|A|^2}\right),
\end{equation}
where ``$dU$'' is the normalized Haar measure on $R \ox S \cong A
\ox B$.
\end{lemma}
\begin{figure}
    \centering
    \includegraphics[scale=0.85]{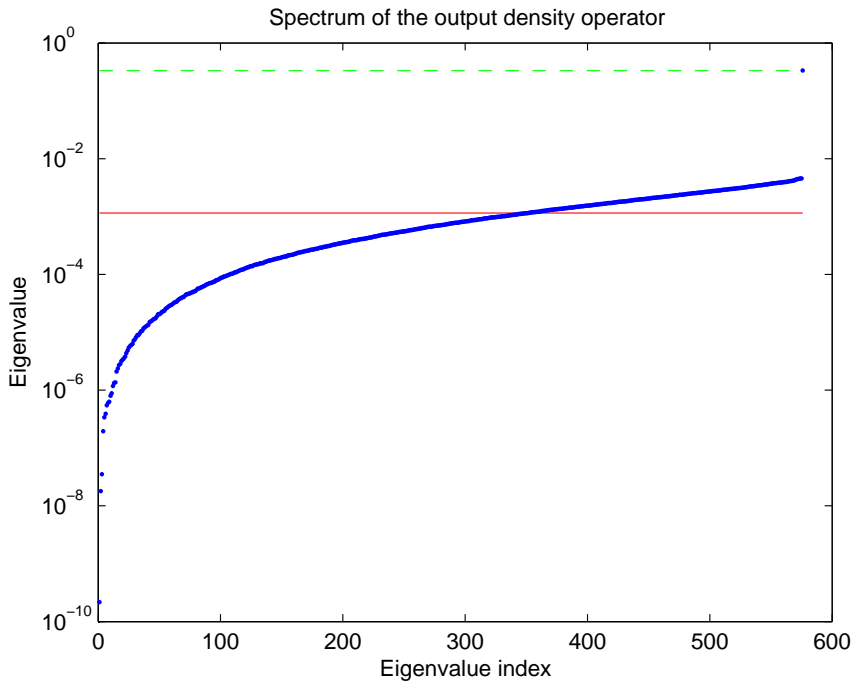}
    % plot_ortho_eig(3,8,24);Ø
    \caption{Typical eigenvalue spectrum of $(\cN \ox \bar\cN)(\Ph)$ when $|R|=3$ and
    $|A|=|B|=24$. The eigenvalues are plotted in increasing order
    from left to right. The green dashed line corresponds to
    $|S|/(|A||B|) = 1/3$, which is essentially equal to the largest
    eigenvalue. The red solid line represents the value
    $(1 - \smfrac{|S|}{|A||B|}) / |A|^2 = 1/864$. If the density operator
    were maximally mixed aside from its largest eigenvalue, all but
    that one eigenvalue would fall on this line. While that is not
    the case here or in general, the remaining eigenvalues are nonetheless sufficiently
    small to ensure that the density operator has high von Neumann
    entropy.}
    \label{fig:eigenspectrum}
\end{figure}
A description of the calculation can be found in Appendix
\ref{appendix:purity}. Let the eigenvalues of $(\cN\ox\bar\cN)(\Ph)$
be equal to $\l_1 \geq \l_2 \geq \cdots \geq \l_{|A|^2}$. For a
typical $U$, Lemmas \ref{lem:big.eigenvalue} and
\ref{lem:avg.purity} together imply that
\begin{equation}
\sum_{j > 1} \l_j^2 = \cO\left(\frac{1}{|A|^2}\right).
\end{equation}
Thus, aside from $\l_1$, the eigenvalues $\l_j$ must be quite small.
A typical eigenvalue distribution is plotted in Figure
\ref{fig:eigenspectrum}. If we define $\tilde\l_j = \l_j / (1 -
\l_1)$, then $\sum_{j>1} \tilde\l_j = 1$ and
\begin{equation}
H_1(\tilde\l)
 \geq H_2(\tilde\l)
 = - \ln \sum_{j>1} \tilde\l_j^2
 = 2 \ln |A| - \cO(1).
\end{equation}
An application of the grouping property then gives us a good lower
bound on the von Neumann entropy:
\begin{equation}
H_1\big((\cN\ox\bar\cN)(\Ph)\big)
 = H_1(\l)
 = h( \l_1 ) + (1 - \l_1) H_1 ( \tilde\l )
 = 2 \ln |A| - \cO(1),
\end{equation}
where $h$ is the binary entropy function. This entropy is nearly as
large as it can be and, in particular, as large as $\Hm( \cN ) +
\Hm(\bar\cN)$ according Theorem IV.1 of \cite{HLW06}, the von
Neumann entropy version of Lemma \ref{lem:size.subspace}.

\section{Discussion} \label{sec:discussion}

The counterexamples presented here demonstrate that the maximal
$p$-norm multiplicativity conjecture is false for $1 < p < 2$. The
primary motivation for studying this conjecture was that it is a
natural strengthening of the minimum entropy output conjecture,
which is of fundamental importance in quantum information theory. In
particular, since the multiplicativity conjecture was formulated,
most attempts to prove the minimum entropy output conjecture for
special cases actually proved maximal $p$-norm multiplicativity and
then took the limit as $p$ decreases to 1. This strategy, we now
know, cannot be used to prove the conjecture in general.

From that perspective, it would seem that the results in this paper
cast doubt on the validity of the minimum entropy output conjecture
itself. However, as we have shown, the examples explored here appear
to be completely consistent with the conjecture, precisely because
the von Neumann entropy is more difficult to perturb than the Renyi
entropies of order $p > 1$. Indeed, the message of this paper may be
that attempts to prove the minimum entropy output conjecture have
all been approaching $p=1$ from the wrong direction. It is quite
possible that additivity of the minimum output Renyi entropy holds
for $0 \leq p \leq 1$ and then fails dramatically for $p > 1$.

This is not, unfortunately, a very well-informed speculation. With
few exceptions~\cite{WolfE05}, there has been very little research
on the additivity question in the regime $p < 1$, even though many
arguments can be easily adapted to this parameter region.
(Eq.~(\ref{eqn:king}), for example, holds for all $0 < p$.)
Remedying this oversight would now seem to be a priority.

{\bf Note:} In the short time since I circulated a preliminary
version of this manuscript, there has already been substantial
further progress. Most notably, Andreas Winter has observed that the
Lipschitz bound in Lemma \ref{lem:size.subspace} is not tight. By
improving it, he managed to disprove the $p$-norm multiplicativity
conjecture at its final redoubt of $p=2$ and simultaneously
demonstrate that the large violation observed here near $p=1$ can be
found for all $p
> 1$. Also, Fr\'ed\'eric Dupuis confirmed that replacing the unitary
group by the orthogonal group leads to qualitatively similar
conclusions. Those and other developments will appear
in~\cite{DHLW07}.

\subsection*{Acknowledgments}
I would like to thank Fr\'ed\'eric Dupuis and Debbie Leung for an
inspiring late-night conversation at the Perimeter Institute,
Andreas Winter for the timely determination that another proposed
class of counterexamples was faulty, Aram Harrow for several
insightful suggestions, and Mary Beth Ruskai for discussions on the
additivity conjecture. I'd also like thank BIRS for their
hospitality during the Operator Structures in Quantum Information
workshop, which rekindled my interest in the additivity problem.
This research was supported by the Canada Research Chairs program, a
Sloan Research Fellowship, CIFAR, FQRNT, MITACS and NSERC.

\appendix

\section{Proof of Lemma \ref{lem:avg.purity}} \label{appendix:purity}

We will estimate the integral, in what is perhaps not the most
illuminating way, by expressing it in terms of the matrix entries of
$U$. Let $U_{s,ab} = \,^{R} \bra{0}^{S}\bra{s} U \ket{a}^A
\ket{b}^B$. Expanding gives
\begin{align} \label{eqn:big.sum}
\int \Tr\Big[ & \big( (\cN \ox \bar{\cN})(\proj{\Ph}) \big)^2 \Big] \, dU \\
 &= \frac{1}{|S|^2}
    \sum_{\stackrel{a_1,a_2}{a_1',a_2'}}
    \sum_{\stackrel{b_1,b_2}{b_1',b_2'}}
    \sum_{\stackrel{s_1,s_2}{s_1',s_2'}}
    \int
    \bar U_{s_1,a_2 b_2} \bar U_{s_2,a_1' b_1}
            \bar U_{s_1',a_2' b_2'} \bar U_{s_2',a_1 b_1'}
            U_{s_1,a_1 b_1} U_{s_2,a_2' b_2}
            U_{s_1',a_1' b_1'} U_{s_2',a_2 b_2'} \, dU. \nonumber
\end{align}
Following \cite{AL03,AL04}, the non-zero terms in the sum can be
represented using a simple graphical notation. Make two parallel
columns of four dots, then label the left-hand dots by the indices
$(s_1, s_2, s_1', s_2')$ and the right-hand dots by the indices
$\vec v = ( a_2 b_2, a_1'b_1, a_2'b_2',a_1b_1')$. Join dots with a
solid line if the corresponding $\bar{U}$ matrix entry appears in
Eq.~(\ref{eqn:big.sum}). Since terms integrate to a non-zero value
only if the vector of $U$ indices $\vec w =
(a_1b_1,a_2'b_2,a_1'b_1',a_2b_2')$ is a permutation of the vector of
$\bar U$ indices, a non-zero integral can be represented by using a
dotted line to connect left-hand and right-hand dots whenever the
corresponding $U$ matrix entry appears in the integral.

Assuming for the moment that the vertex labels in the left column
are all distinct and likewise for the right column, the integral
evaluates to the Weingarten function $\Wg(\pi)$, where $\pi$ is the
permutation such that $w_i = v_{\pi(i)}$. For the rough estimate
required here, it is sufficient to know that $\Wg(\pi) = \Theta\big(
(|A||B|)^{-4-|\pi|} \big)$, where $|\pi|$ is the minimal number of
factors required to write $\pi$ as a product of transpositions, and
that $\Wg(e) = (|A||B|)^{-4}\big( 1 +
\cO(|A|^{-2}|B|^{-2})\big)$~\cite{CS06}.

The dominant contribution to Eq.~(\ref{eqn:big.sum}) comes from the
``stack'' diagram
\[
\xy
 (0,0)*{\bullet}="l4";  (0,7)*{\bullet}="l3";
 (0,14)*{\bullet}="l2"; (0,21)*{\bullet}="l1";
 (13,0)*{\bullet}="r4"; (13,7)*{\bullet}="r3";
 (13,14)*{\bullet}="r2";(13,21)*{\bullet}="r1";
 "l1" ; "r1" **\dir{-}; "l2" ; "r2" **\dir{-};
 "l3" ; "r3" **\dir{-}; "l4" ; "r4" **\dir{-};
 (-4,21)*{s_1};  (-4,14)*{s_2};
 (-4,7)*{s_1'}; (-4,0)*{s_2'};
 (25,21)*{a_2 b_2 = a_1 b_1};
 (25,14)*{a_1' b_1 = a_2' b_2};
 (25,7)*{a_2' b_2' = a_1' b_1'};
 (25,0)*{a_1 b_1' = a_2 b_2',};
\endxy
\]
in which the solid and dashed lines are parallel and for which the
contribution is positive and approximately equal to
\begin{equation} \label{eqn:dominant}
 \frac{1}{|S|^2}
    \sum_{\stackrel{a_1,a_2}{a_1',a_2'}}
    \sum_{\stackrel{b_1,b_2}{b_1',b_2'}}
    \sum_{\stackrel{s_1,s_2}{s_1',s_2'}}
    \d_{a_1a_2}\d_{b_1b_2}\d_{a_1'a_2'}\d_{b_1'b_2'}
    \Wg(\id)
% = \frac{1}{|S|^2} |A|^2 |B|^2
%    \Theta\left(\frac{1}{|A|^4|B|^4}\right)
 = \frac{|S|^2}{|A|^2|B|^2}\left( 1 + \cO(|A|^{-2}|B|^{-2} ) \right).
\end{equation}
(The expression on the left-hand side would be exact but for the
terms in which vertex labels are not distinct.) To obtain an
estimate of Eq.~(\ref{eqn:big.sum}), it is then sufficient to
examine the other terms and confirm that they are all of smaller
asymptotic order than this. There are six diagrams representing
transpositions, and their associated (negative) contributions are
\[
\xy
 (0,0)*{\xy
     (0,0)*{\bullet}="l4";  (0,7)*{\bullet}="l3";
     (0,14)*{\bullet}="l2"; (0,21)*{\bullet}="l1";
     (13,0)*{\bullet}="r4"; (13,7)*{\bullet}="r3";
     (13,14)*{\bullet}="r2";(13,21)*{\bullet}="r1";
     "l1" ; "r1" **\dir{-}; "l2" ; "r2" **\dir{-};
     "l3" ; "r3" **\dir{-}; "l4" ; "r4" **\dir{-};
     (-4,21)*{s_1};  (-4,14)*{s_2};
     (-4,7)*{s_1'}; (-4,0)*{s_2'};
     (25,21)*{a_2 b_2 = a_1 b_1};
     (25,14)*{a_1' b_1 = a_2' b_2};
     (25,7)*{a_2' b_2' = a_2 b_2'};
     (25,0)*{a_1 b_1' = a_1' b_1'};
     "l3"; "r4" **\dir{--};
     "l4"; "r3" **\dir{--};
     (15,-7)*{\Theta( |S|^2|A|^{-4}|B|^{-2})}
     \endxy};
 (50,0)*{\xy
     (0,0)*{\bullet}="l4";  (0,7)*{\bullet}="l3";
     (0,14)*{\bullet}="l2"; (0,21)*{\bullet}="l1";
     (13,0)*{\bullet}="r4"; (13,7)*{\bullet}="r3";
     (13,14)*{\bullet}="r2";(13,21)*{\bullet}="r1";
     "l1" ; "r1" **\dir{-}; "l2" ; "r2" **\dir{-};
     "l3" ; "r3" **\dir{-}; "l4" ; "r4" **\dir{-};
     (-4,21)*{s_1};  (-4,14)*{s_2};
     (-4,7)*{s_1'}; (-4,0)*{s_2'};
     (25,21)*{a_2 b_2 = a_1 b_1};
     (25,14)*{a_1' b_1 = a_2 b_2' };
     (25,7)*{a_2' b_2' = a_1' b_1'};
     (25,0)*{a_1 b_1' = a_2' b_2};
     "l2"; "r4" **\dir{--};
     "l4"; "r2" **\dir{--};
     (15,-7)*{\Theta( |S|^2|A|^{-4}|B|^{-4})}
     \endxy};
 (100,0)*{\xy
     (0,0)*{\bullet}="l4";  (0,7)*{\bullet}="l3";
     (0,14)*{\bullet}="l2"; (0,21)*{\bullet}="l1";
     (13,0)*{\bullet}="r4"; (13,7)*{\bullet}="r3";
     (13,14)*{\bullet}="r2";(13,21)*{\bullet}="r1";
     "l1" ; "r1" **\dir{-}; "l2" ; "r2" **\dir{-};
     "l3" ; "r3" **\dir{-}; "l4" ; "r4" **\dir{-};
     (-4,21)*{s_1};  (-4,14)*{s_2};
     (-4,7)*{s_1'}; (-4,0)*{s_2'};
     (25,21)*{a_2 b_2 = a_2 b_2'};
     (25,14)*{a_1' b_1 = a_2' b_2};
     (25,7)*{a_2' b_2' = a_1' b_1'};
     (25,0)*{a_1 b_1' = a_1 b_1,};
     "l1"; "r4" **\dir{--};
     "l4"; "r1" **\dir{--};
     (15,-7)*{\Theta( |S|^2|A|^{-2}|B|^{-4})}
     \endxy}
\endxy
\]
\[
\xy
 (0,0)*{\xy
     (0,0)*{\bullet}="l4";  (0,7)*{\bullet}="l3";
     (0,14)*{\bullet}="l2"; (0,21)*{\bullet}="l1";
     (13,0)*{\bullet}="r4"; (13,7)*{\bullet}="r3";
     (13,14)*{\bullet}="r2";(13,21)*{\bullet}="r1";
     "l1" ; "r1" **\dir{-}; "l2" ; "r2" **\dir{-};
     "l3" ; "r3" **\dir{-}; "l4" ; "r4" **\dir{-};
     (-4,21)*{s_1};  (-4,14)*{s_2};
     (-4,7)*{s_1'}; (-4,0)*{s_2'};
     (25,21)*{a_2 b_2 = a_1 b_1};
     (25,14)*{a_1' b_1 = a_1' b_1'};
     (25,7)*{a_2' b_2' = a_2' b_2};
     (25,0)*{a_1 b_1' = a_2 b_2'};
     "l3"; "r2" **\dir{--};
     "l2"; "r3" **\dir{--};
     (15,-7)*{\Theta( |S|^2|A|^{-2}|B|^{-4})}
     \endxy};
 (50,0)*{\xy
     (0,0)*{\bullet}="l4";  (0,7)*{\bullet}="l3";
     (0,14)*{\bullet}="l2"; (0,21)*{\bullet}="l1";
     (13,0)*{\bullet}="r4"; (13,7)*{\bullet}="r3";
     (13,14)*{\bullet}="r2";(13,21)*{\bullet}="r1";
     "l1" ; "r1" **\dir{-}; "l2" ; "r2" **\dir{-};
     "l3" ; "r3" **\dir{-}; "l4" ; "r4" **\dir{-};
     (-4,21)*{s_1};  (-4,14)*{s_2};
     (-4,7)*{s_1'}; (-4,0)*{s_2'};
     (25,21)*{a_2 b_2 = a_1' b_1'};
     (25,14)*{a_1' b_1 = a_2' b_2 };
     (25,7)*{a_2' b_2' = a_1 b_1};
     (25,0)*{a_1 b_1' = a_2 b_2'};
     "l1"; "r3" **\dir{--};
     "l3"; "r1" **\dir{--};
     (15,-7)*{\Theta( |S|^2|A|^{-4}|B|^{-4})}
     \endxy};
 (100,0)*{\xy
     (0,0)*{\bullet}="l4";  (0,7)*{\bullet}="l3";
     (0,14)*{\bullet}="l2"; (0,21)*{\bullet}="l1";
     (13,0)*{\bullet}="r4"; (13,7)*{\bullet}="r3";
     (13,14)*{\bullet}="r2";(13,21)*{\bullet}="r1";
     "l1" ; "r1" **\dir{-}; "l2" ; "r2" **\dir{-};
     "l3" ; "r3" **\dir{-}; "l4" ; "r4" **\dir{-};
     (-4,21)*{s_1};  (-4,14)*{s_2};
     (-4,7)*{s_1'}; (-4,0)*{s_2'};
     (25,21)*{a_2 b_2 = a_2' b_2};
     (25,14)*{a_1' b_1 = a_1 b_1};
     (25,7)*{a_2' b_2' = a_1' b_1'};
     (25,0)*{a_1 b_1' = a_2 b_2'};
     "l1"; "r2" **\dir{--};
     "l2"; "r1" **\dir{--};
     (15,-7)*{\Theta( |S|^2|A|^{-4}|B|^{-2}).}
     \endxy}
\endxy
\]
For permutations $\pi$ such that $|\pi| > 1$, the Weingarten
function is significantly suppressed: $\Wg(\pi) = \cO( |A|^{-6}
|B|^{-6})$. Moreover, for a given diagram type, the requirement that
$w_i = v_{\pi(i)}$ can only hold if at least two pairs of the
indices $a_1,a_2,a_1',a_2',b_1,b_2,b_1',b_2'$ are identical. The
contribution from such diagrams is therefore
$\cO(|S|^2|A|^{-4}|B|^{-2})$.

To finish the proof, it is necessary to consider integrals in which
the vertex labels on the left- or the right-hand side of a diagram
are not all distinct. In this more general case, choosing a set
$\cC$ of representatives for the conjugacy classes of the
permutation group on four elements, the value of the integral can be
written
\begin{equation} \label{eqn:cycle.decomp}
 \sum_{c \in \cC} N(c) \Wg(c),
\end{equation}
where
\begin{equation}
 N(c) = \sum_{\stackrel{\s \in \cS_4:}{\vec v = \s(\vec v)}}
        \sum_{\stackrel{\tau \in \cS_4:}{\vec w = \tau(\vec w)}}
        \d( \tau \pi \s \in c ).
\end{equation}
These formulas have a simple interpretation. Symmetry in the vertex
labels introduces ambiguities in the diagrammatic notation; the
formula states that every one of the diagrams consistent with a
given vertex label set must be counted, and with a defined
dimension-independent multiplicity. Conveniently, our crude
estimates have already done exactly that, ignoring the
multiplicities. The only case for which we need to know the
multiplicities, moreover, is for contributions to the dominant term,
which we want to know exactly and not just up to a constant
multiple.

We claim that in the sum (\ref{eqn:big.sum}) there are at most
$\cO(|S|^4|A||B|^3)$ terms with vertex label symmetry. The total
contribution for terms with vertex label symmetries $\tau$ and $\s$
in which $|\tau \pi \s| \geq 1$ is therefore of size
$\cO(|S|^2|A|^{-4}|B|^{-2})$ and does not affect the dominant term.
To see why the claim holds, fix a diagram type and recall that the
requirement $w_i = v_{\pi(i)}$ for a permutation $\pi$ can only hold
if at least two pairs of the indices
$a_1,a_2,a_1',a_2',b_1,b_2,b_1',b_2'$ are identical. Equality is
achieved only when all the $A$ indices or all the $B$ indices are
aligned, corresponding to the following two diagrams:
\[
\xy
 (0,0)*{\xy
     (0,0)*{\bullet}="l4";  (0,7)*{\bullet}="l3";
     (0,14)*{\bullet}="l2"; (0,21)*{\bullet}="l1";
     (13,0)*{\bullet}="r4"; (13,7)*{\bullet}="r3";
     (13,14)*{\bullet}="r2";(13,21)*{\bullet}="r1";
     "l1" ; "r1" **\dir{-}; "l2" ; "r2" **\dir{-};
     "l3" ; "r3" **\dir{-}; "l4" ; "r4" **\dir{-};
     (-4,21)*{s_1};  (-4,14)*{s_2};
     (-4,7)*{s_1'}; (-4,0)*{s_2'};
     (25,21)*{a_2 b_2 = a_2' b_2};
     (25,14)*{a_1' b_1 = a_1 b_1};
     (25,7)*{a_2' b_2' = a_2 b_2'};
     (25,0)*{a_1 b_1' = a_1' b_1'};
     "l1"; "r2" **\dir{--};
     "l2"; "r1" **\dir{--};
     "l3"; "r4" **\dir{--};
     "l4"; "r3" **\dir{--};
     \endxy};
 (50,0)*{\xy
     (0,0)*{\bullet}="l4";  (0,7)*{\bullet}="l3";
     (0,14)*{\bullet}="l2"; (0,21)*{\bullet}="l1";
     (13,0)*{\bullet}="r4"; (13,7)*{\bullet}="r3";
     (13,14)*{\bullet}="r2";(13,21)*{\bullet}="r1";
     "l1" ; "r1" **\dir{-}; "l2" ; "r2" **\dir{-};
     "l3" ; "r3" **\dir{-}; "l4" ; "r4" **\dir{-};
     (-4,21)*{s_1};  (-4,14)*{s_2};
     (-4,7)*{s_1'}; (-4,0)*{s_2'};
     (25,21)*{a_2 b_2 = a_2 b_2'};
     (25,14)*{a_1' b_1 = a_1' b_1' };
     (25,7)*{a_2' b_2' = a_2' b_2};
     (25,0)*{a_1 b_1' = a_1 b_1};
     "l1"; "r4" **\dir{--};
     "l4"; "r1" **\dir{--};
     "l2"; "r3" **\dir{--};
     "l3"; "r2" **\dir{--};
     \endxy};
\endxy
\]
For the first diagram, using the fact that $|A|\leq|B|\leq|S|$, it
is easy to check that imposing the extra constraint that either the
top or bottom two $S$ or $AB$ vertex labels match singles at most
$\cO(|S|^4|A||B|^3)$ terms from Eq.~(\ref{eqn:big.sum}). Similar
reasoning applies to the second diagram, but imposing the constraint
instead on rows one and four, or two and three. For all other
diagram types, at least four pairs of the indices
$a_1,a_2,a_1',a_2',b_1,b_2,b_1',b_2'$ are identical. (The number of
matching $A$ and $B$ indices is necessarily even.) In a term for
which the vertex labels are not all distinct, either a pair of $S$
indices or a further pair of $A$ or $B$ indices must be identical.
In the latter case, there must exist an identical $A$ pair
\emph{and} an identical $B$ pair among all the pairs. Again using
$|A| \leq |B| \leq |S|$, there can be at most $\cO(|S|^4|B|^3)$ such
terms per diagram type, which demonstrates the claim.

We are thus left to consider integrals with vertex label symmetry
and $N(e) \neq 0$ in Eq.~(\ref{eqn:cycle.decomp}). If $N(e) = 1$,
then our counting was correct and there is no problem. It is
therefore sufficient to bound the number of integrals in which $N(e)
> 1$. This can occur only in terms with at least 2 vertex label
symmetries. Running the argument of the previous paragraph again,
for the two diagrams with $A$ or $B$ indices all aligned, this
occurs in at most $\cO(|S|^4|B|^2)$ terms. For the rest of the
cases, it is necessary to impose equality on yet another pair of
indices, leading again to at most $\cO(|S|^4|B|^2)$ terms. Since
$\Wg(e) = \cO(|A|^{-4}|B|^{-4})$, these contributions are
collectively $\cO(|S|^2|A|^{-4}|B|^{-2})$.

The bound on the error term in Eq.~(\ref{eqn:avg.purity}) arises by
substituting the inequalities $|S| \leq |A||B|$ and $|A| \leq |B|$
into each of the estimates calculated above.

\bibliographystyle{unsrt}
\bibliography{add}

\end{document}